\documentclass[sigconf]{acmart}

\usepackage{booktabs} 
\usepackage{bm}
\usepackage{wrapfig}
\usepackage{float}
\usepackage{array}
\usepackage{subcaption}
\usepackage{graphicx}
\newcolumntype{L}[1]{>{\raggedright\let\newline\\\arraybackslash\hspace{0pt}}m{#1}}
\usepackage{centernot}
\usepackage{enumitem}
\usepackage[ruled, algo2e, linesnumbered]{algorithm2e}
\usepackage{mathtools}
\DeclarePairedDelimiter{\ceil}{\lceil}{\rceil}

\setcopyright{rightsretained}
 
 \newcommand{\attripart}{\textsc{AttriPart}}
 \newcommand{\localforecasting}{\textsc{LocalForecasting}}
 \newcommand{\localproximity}{\textsc{LocalProximity}}
 \newcommand{\expandedneighborhood}{\textsc{ExpandedNeighborhood}}

\newcommand{\hide}[1]{}

\acmDOI{10.475/123_4}

\acmISBN{123-4567-24-567/08/06}

\acmConference[KDD'18]{ACM KDD}{August 2018}{London, United Kingdom}
\acmYear{2018}
\copyrightyear{2018}

\acmArticle{4}
\acmPrice{15.00}

\begin{document}
\title{Local Partition in Rich Graphs}

\author{Scott Freitas}
\affiliation{%
  \institution{Arizona State University}
  \city{Tempe}
  \state{Arizona}
  \postcode{85028}
}
\email{scott.freitas@asu.edu}

\author{Hanghang Tong}
\affiliation{%
  \institution{Arizona State University}
  \city{Tempe}
  \state{Arizona}
  \postcode{85028}
}
\email{hanghang.tong@asu.edu}

\author{Nan Cao}
\affiliation{
  \institution{Tongji University}
  \city{Shanghai} 
  \state{China} 
}
\email{nan.cao@gmail.com}

\author{Yinglong Xia}
\affiliation{
  \institution{Huawei}
  \city{Santa Clara} 
  \state{California} 
  }
\email{yinglong.xia@huawei.com}

\renewcommand{\shortauthors}{S. Freitas et al.}

\begin{abstract} Local graph partitioning is a key graph mining tool that allows researchers to identify small groups of interrelated nodes (e.g. people) and their connective edges (e.g. interactions). Because local graph partitioning is primarily focused on the network structure of the graph (vertices and edges), it often fails to consider the additional information contained in the attributes. In this paper we propose---(i) a scalable algorithm to improve local graph partitioning by taking into account \textit{both} the network structure of the graph and the attribute data and (ii) an application of the proposed local graph partitioning algorithm (\attripart) to predict the evolution of local communities (\localforecasting). Experimental results show that our proposed \attripart\ algorithm finds up to \textbf{1.6$\times$} denser local partitions, while running approximately \textbf{43$\times$} faster than traditional local partitioning techniques (PageRank-Nibble \cite{PageRank-Nibble}). In addition, our \localforecasting\ algorithm shows a significant improvement in the number of nodes and edges correctly predicted over baseline methods.

\end{abstract}

%
%



\maketitle

\section{Introduction}
\textbf{Motivation.} With the rise of the big data era an exponential amount of network data is being generated at an unprecedented rate across many disciplines. One of the critical challenges before us is the translation of this large-scale network data into meaningful information. A key task in this translation is the identification of \textit{local communities} with respect to a given seed node \footnote{we interchangeably refer to local community as a local partition}. In practical terms, the information discovered in these local communities can be utilized in a wide range of high-impact areas---from the micro (protein interaction networks \cite{Protein1} \cite{Protein2}) to the macro (social \cite{Social1} \cite{LocalSocial} and transportation networks). 

\textbf{Problem Overview.} How can we quickly determine the local graph partition around a given seed node? This problem is traditionally solved using an algorithm like Nibble \cite{Nibble}, which identifies a small cluster in time proportional to the size of the cluster, or PageRank-Nibble, \cite{PageRank-Nibble} which improves the running time and approximation ratio of Nibble with a smaller polylog time complexity. While both of these methods provide powerful techniques in the analysis of network structure, they fail to take into account the attribute information contained in many real-world graphs. Other techniques to find improved rank vectors, such as attributed PageRank \cite{AttriRank}, lack a generalized conductance metric for measuring cluster "goodness" containing attribute information. In this paper, we propose a novel method that combines the network structure and attribute information contained in graphs---to better identify local partitions using a generalized conductance metric.


\textbf{Applications.} Local graph partition plays a central role in many application scenarios. For example, a common problem in \textit{recommender systems} is that of social media networks and determining how a local community will evolve over time. The proposed \localforecasting\ algorithm can be used to determine the evolution of local communities, which can then assist in user recommendations. Another example utilizing social media networks is {\em ego-centric network identification}, where the goal is to identify the locally important neighbors relative to a given person. To this end, we can use our \attripart\ algorithm to identify better ego-centric networks using the graph's network structure and attribute information. Finally, newly arrived nodes (i.e., {\em cold-start nodes}) often contain few connections to their surrounding neighbors, making it difficult to ascertain their grouping to various communities. The proposed \localforecasting\ algorithm mitigates this problem by introducing additional attribute edges (link prediction), which can assist in determining which local partitions the cold start nodes will belong to in the future.

\hide{
There are multiple applications of this research in the areas of (1) recommender systems, (2) ego-centric network identification and (3) cold start scenarios. 

}
\textbf{Contributions.} Our primary contributions are three-fold: 
\begin{itemize}
    \item The formulation of a graph model and generalized conductance metric that incorporates both attribute and network structure edges. 
    \item The design and analysis of local clustering algorithm \attripart\ and local community prediction algorithm \localforecasting. Both algorithms utilize the proposed graph model, modified conductance metric and novel subgraph identification technique.
    \item The evaluation of the proposed algorithms on three real-world datasets---demonstrating the ability to rapidly identify denser local partitions compared to traditional techniques.
\end{itemize}

\textbf{Deployment.} The local partitioning algorithm \attripart\ is currently \textbf{deployed} to the PathFinder \cite{RapidAnalysisNetworkConnectivity} web platform (\url{www.path-finder.io}), with the goal of assisting users in mining local network connectivity from large networks. The design and deployment challenges were wide ranging, including---(i) the integration of four different programming languages, (ii) obtaining real-time performance with low cost hardware and (iii) implementation of a visually appealing and easy to use interface. We note that the \attripart\ algorithm, deployed to the web platform, has \textit{performance nearly identical} to the results presented in section \ref{ExperimentsAndDiscussion}.

\begin{figure}[H]
\includegraphics[height=5.27cm, width=8.5cm]{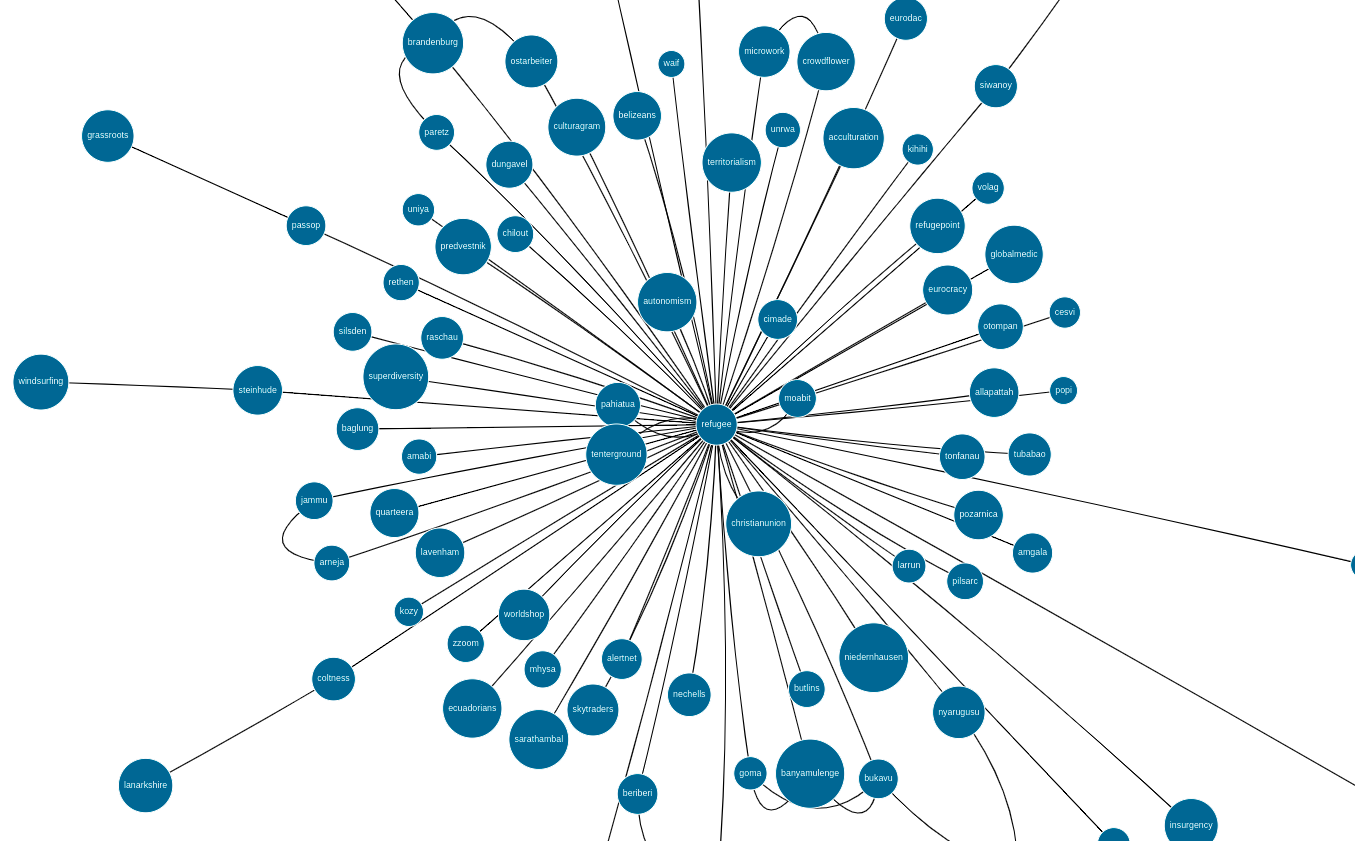}
\caption{Close-up of the \attripart\ algorithm on the PathFinder web platform.}
\label{PlatformPicture}
\end{figure}

This paper is organized as follows---Section 2 defines the problem of local partitioning in rich graphs; Section 3 introduces our proposed model and algorithms; Section 4 presents our experimental results on multiple real-world datasets; Section 5 reviews the related literature; and Section 6 concludes the paper.


\section{Problem Definition}\label{GraphDefs}
In this paper we consider three graphs---(1) an undirected, unweighted structure graph $\textbf{G} = (V, E)$, (2) an undirected, weighted attribute graph $\textbf{A} = (V, E)$ and (3) a combined graph consisting of both $\textbf{G}$ and $\textbf{A}$ that is undirected and weighted $\textbf{B} = (V, E)$. In each graph, $V$ is the set of vertices, $E$ is the set of edges, $n$ is the number of vertices and $m$ is the number of edges (i.e. $G$, $H$ and $B$ contain the same number of vertices and edges by default). In order to denote the degree centrality we say $\delta(v)$ is the degree of vertex $v$. We use bold uppercase letters to denote matrices (e.g. \textbf{G}) and bold lowercase letters to denote vectors (e.g. \textbf{v}).


For the ease of description, we define terms that are interchangeably used throughout the literature and this paper---(a) we refer to network as a graph, (b) node is synonymous with vertex, (c) local partition is referred to as a local cluster, (d) seed node is equivalent to query and start vertex, (e) topological edges of the graph refers to the network structure of the graph, (f) a rich graph is a graph with attributes on the nodes and or edges.

Having outlined the notation, we define the problem of local partitioning in rich graphs as follows:

\vspace{0.25cm}
\textit{Problem 1. Local Partitioning in Rich Graphs}

\noindent
\textit{\textbf{Given:} (1) an undirected, unweighted graph $\textbf{G} = (V,E)$, (2) a seed node $q \in V$ and (3) attribute information for each node $v \in V$ containing a k-dimensional attribute vector \bm{$x_i$}---with an attribute matrix $\textbf{X} = [\bm{x_1, x_2, ..., x_n}] \in \mathbb{R}^{k \times n}$ representing the attribute vector for each node $v$.}

\noindent
\textit{\textbf{Output:} a subset of vertices $S \subset V$ such that $S$ best represents the local partition around seed node $q$ in graph \bm{$B$}.}


\vspace{-0.2cm}
\begin{table}[h!]
\caption{\large{Symbols and Definition}}
\label{table:1}
\vspace{-0.35cm}
\begin{center}
 \begin{tabular}{|c | L{6.5cm}|} 
 \hline
 \textbf{Symbol} & \textbf{Definition} \\ [0.5ex] 
 \hline
 \bm{$G$}, \bm{$A$}, \bm{$B$} & network, attribute \& combined graphs \\ 


 
 $n$, $m$ & number of nodes \& edges in graphs \bm{$G$}, \bm{$A$}, \bm{$B$} \\
 
 
 $m_e$ & number of edges in \bm{$B$} after \localforecasting \\
 
 $p$, $m_p$ & number of nodes \& edges in \bm{$T$} \\
 \hline\hline
 
 \bm{$s$}, $q$, $\phi_o$ & preference vector, seed node \& target conductance \\

 \bm{$W$} & lazy random walk transition matrix \\


 $S$ & set of vertices representing local partition \\


 $\epsilon$, $\epsilon_t$ & rank truncation and iteration thresholds \\

 $t_m$, $n_s$ & rank vector iterations; number of vertices to sweep  \\

 $\alpha_n$, $\alpha_r$ & \attripart\ \& \localproximity\ teleport values \\

 \hline\hline
 $t_s$, $n_w$ & subgraph relevance threshold \& number of walks \\
 
 
 
 
 \bm{$T$}; $D$, $L$ & subgraph of \bm{$B$}; walk count dictionary \& list  \\ 


 $\mu(L)$, $\sigma(L)$ & mean and standard deviation of $L$  \\

 
 \hline\hline
 $t_e$ & edge addition threshold \\
 \hline
\end{tabular}
\end{center}
\end{table}

\section{Methodology}
This section first describes the preliminaries for our proposed algorithms, including the graph model and modified conductance metric. Next, we introduce each proposed algorithm---(1) \localproximity, (2) \attripart\ and (3) \localforecasting. Finally, we provide an analysis of the proposed algorithms in terms of effectiveness and efficiency. 

\subsection{Preliminaries}
\paragraph{{Graph Model.}} Topological network \bm{$G$} represents the network structure of the graph and is formally defined in Eq.~\eqref{Graph_G}. Attribute network \bm{$A$} represents the attribute structure of the graph and is computed based on the similarity for every edge  $(u, v) \in E$ in \bm{$G$}. In order to determine the similarity between the two nodes, we use Jaccard Similarity $J(u, v)$. 
\bm{$A$} is formally defined in Eq.~\eqref{Graph_A} where 0.05 is the default attribute similarity between an edge $(u, v) \in E$ in \bm{$G$} if $J(x_u, x_v) = 0$. In addition, $t_e$ is the similarity threshold for the addition of edges not in \bm{$G$} where $0 < t_e \leq 1$. Combined Network \bm{$B$} represents the combined graph of \bm{$G$} and \bm{$A$} and is formally defined in Eq.~\eqref{Graph_B}.

\begin{figure} 
\includegraphics[height=3.01cm, width=8.5cm]{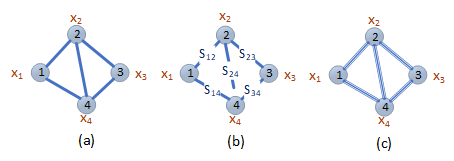}
\caption{Example of the three graph models: (a) graph \bm{$G$} is the network structure with nodes $\{1, 2, 3, 4\}$ and corresponding attribute set $\{\bm{x_1, x_2, x_3, x_4}\}$ given as input. (b) Graph \bm{$A$} is the attribute network with the same set of edges as \bm{$G$} with each edge $(u, v)$ assigned a positive similarity weight $s_{uv}$. (c) Graph \bm{$B$} is a linear combination of the each respective edge $(u, v)$ from \bm{$G$} and \bm{$A$}.}
\label{GraphModels}
\end{figure}

Formally, we define each of the three graph models \bm{$G$}, \bm{$A$} and \bm{$B$} in Eq.~\eqref{Graph_G}, Eq.~\eqref{Graph_A} and Eq.~\eqref{Graph_B}. Figure \ref{GraphModels} presents an illustrative example.

\begin{equation} \label{Graph_G}
 G(u,v) =
    \begin{cases}
      1, & \text{if $(u,v) \in E$ and } $u$ \neq $v$ \\
      0, & \text{otherwise}
    \end{cases}
\end{equation}

\begin{equation} \label{Graph_A}
    A(u,v) =
    \begin{cases}
      $J(u,v)$, & \text{if $(u,v) \in E$, }\ $u$ \neq $v$ \text{ and $J(u,v) > 0$} \\
      0.05, & \text{if $(u,v) \in E$, }\ $u$ \neq $v$ \text{ and $J(u,v) = 0$} \\
      $J(u,v)$, & \text{if $(u,v) \not\in E$, }\ $u$ \neq $v$ \text{ and $J(u,v) > t_e$} \\
      0, & \text{otherwise}
    \end{cases}
\end{equation}

\begin{equation} \label{Graph_B}
    B(u,v) =
    \begin{cases}
      1 + $A(u,v)$, & \text{if $(u,v) \in E$ and $(u,v) \in A$} \\
      $A(u,v)$, & \text{if $(u,v) \not\in E$ and $(u,v) \in A$} \\
      0, & \text{otherwise}
    \end{cases}
\end{equation}





\paragraph{\textbf{Conductance}}

Conductance is a standard metric for determining how tight knit a set of vertices are in a graph \cite{Conductance}. The traditional conductance metric is defined in Eq.~\eqref{TraditionalConductance}, where $S$ is the set of vertices representing the local partition. The lower the conductance value $\phi(S)$, where $0 \leq \phi(S) \leq 1$, the more likely $S$ represents a good partition of the graph.

\begin{equation} \label{TraditionalConductance}
    \phi(S) = \frac{cut(S)}{min(vol(S), vol(\bar{S}))}
\end{equation}

\noindent Where the cut is $Cut(S) = \{(u,v) \in E | u \in S, v \notin S\}$, and the volume is $vol(S) = \sum\limits_{v \in S} \delta(v)$.

\hide{
\begin{equation}
    Cut(S) = \{(u,v) \varepsilon E | u \varepsilon S, v \centernot\varepsilon S\}
\end{equation}

and the volume is

\begin{equation}\label{Volume}
    vol(S) = \sum\limits_{v \epsilon S} \delta(v)
\end{equation}
}

This definition of conductance will serve as the benchmark to compare the results of our parallel conductance metric. 

\textit{Parallel Conductance.} We propose a parallel conductance metric which takes into account both the attribute and topological edges in the graph. Instead of simply adding the cut of each vertex $v \in S$, we want to determine whether $v$ is more similar to the vertices in $S$ or $\bar{S}$.  The new cut and conductance metric is formally defined in Eq.~\eqref{ModifiedCut} and Eq.~\eqref{ModifiedConductance}, respectively. The \textit{key} idea behind the parallel conductance metric is to determine whether each vertex in $S$ is more similar to $S$ or $\bar{S}$ using the additional information provided by the attribute links. 

\begin{equation} \label{ModifiedCut}
\begin{aligned}
    parallel\_cut(S) &= \sum\limits_{i \epsilon S}\frac{ \sum\limits_{j \not\epsilon S} B(i, j)}{ \sum\limits_{j \epsilon S} B(i, j)} = \sum\limits_{i \epsilon S}\frac{ \sum\limits_{j \not\epsilon S} \big[A(i, j) + G(i, j)\big]}{ \sum\limits_{j \epsilon S} \big[A(i, j) + G(i, j)\big]}\\
\end{aligned}
\end{equation}

By definition, \bm{$B$} can be split into its representative components, \bm{$G$} and \bm{$A$}. We also note a few \textit{key} properties of the parallel cut metric below:

\begin{enumerate}
    \item $Parallel\_cut = 1$ means that the vertices in $S$ have connections of equal weighting between $S$ and $\bar{S}$.
    \item $Parallel\_cut < 1$ means that the vertices in $S$ have only a few strong connections to $\bar{S}$.
    \item $Parallel\_cut > 1$ means that the vertices in $S$ are more strongly connected to $\bar{S}$ than $S$.
\end{enumerate}

Eq.~\eqref{ModifiedConductance} uses the cut as defined in Eq.~\eqref{ModifiedCut} and the volume as defined above with the modification that $\delta(v)$ is a sum of it's components in \bm{$G$} and \bm{$A$}. 

\begin{equation} \label{ModifiedConductance}
\begin{aligned}
    \phi(S) &= \frac{parallel\_cut(S)}{vol(S)} \\
\end{aligned}
\end{equation}

We note that the parallel conductance metric has a different scale compared to the traditional conductance metric. For example, a conductance of 0.3 in the traditional conductance doesn't have the same meaning as a conductance of 0.3 in the parallel definition. We also bound the volume of $S$ to $vol(S) <1/2vol(B)$. This allows us to reduce the $min(vol(S), vol(\bar{S}))$ computation to $vol(S)$.

\begin{figure}[h!]
\includegraphics[height=5cm, width=7.5cm]{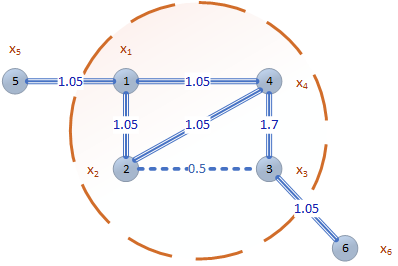}
\caption{A toy example calculating the parallel cut and conductance with local partition $S$ containing vertices $\{1, 2, 3, 4\}$. Parallel cut($V_1$) = 1.05/2.1 = 0.5, parallel cut($V_2$) = 0, parallel cut($V_3$) = 1.05/2.2 = 0.477, parallel cut($V_4$) = 0, parallel cut($Total$) = 0.5 + 0.477 = 0.977. Volume($S$) = 12. Parallel conductance($S$) = 0.977/12 = 0.0814.}
\end{figure}

\subsection{Algorithms}
We propose three algorithms in this subsection, including (1) \localproximity\, (2) \attripart\ and (3) \localforecasting. First, we introduce the \localproximity\ algorithm as a key building block for speeding-up the \attripart\ and \localforecasting\ algorithms by finding a subgraph containing only the nodes and edges relevant to the given seed node. Based on \localproximity, we further propose the \attripart\ algorithm to find a local partition around a seed node by minimizing the parallel conductance metric. Finally, we propose the \localforecasting\ algorithm, which builds upon \attripart, to predict a local community's evolution. 

\textbf{\localproximity.} There are two primary purposes for the \localproximity\ algorithm---(i) the requisite computations for the \localforecasting\ algorithm require a pairwise similarity calculation of all nodes, which is intractable for large graphs due to the quadratic run time. To make this computation feasible, we use the \localproximity\ algorithm to determine a small subgraph of relevant vertices around a given seed node $q$. (ii) We experimentally found that the PageRank vector utilized in the \attripart\ algorithm is significantly faster to compute after running the proposed \localproximity\ algorithm.

\textit{Algorithm Details.}
The goal is to find a subgraph $\bm{T}$ around seed node $q$, such that \bm{$T$} contains only nodes and edges likely to be reached in $n_w$ trials of random walk with restart. We base the importance of a vertex $v \in V$ on the theory that random walks can measure the importance of nodes and edges in a graph \cite{RandomWalkSubgraph}\cite{RandomWalkBetweenness}. This is done by defining node relevance proportional to the frequency of times a random walk with restart walks on a vertex in $n_w$ trials (nodes walked on more than once in a walk will still count as one). Instead of using a simple threshold parameter to determine node/edge relevance as in \cite{RandomWalkSubgraph}, we utilize the mean and standard deviation of the walk distribution in order for the results to remain insensitive of $n_w$ given that $n_w$ is sufficiently large. In conjunction with the mean and standard deviation, we introduce $t_s$ as a relevance threshold parameter to determine the size of the resulting subgraph \bm{$T$}. See section \ref{Algorithm_Analysis} for more details.

\textit{Algorithm Description.}
The \localproximity\ algorithm takes a graph \bm{$B$}, a seed node $q \in \bm{B}$, a teleport value $\alpha_r$, the number of walks to simulate $n_w$, a relevance threshold $t_s$---and returns a subgraph \bm{$T$} containing the relevant vertices in relation to $q$. This algorithm can be viewed in three major steps:

\begin{enumerate}
    \item Compute the walk distribution around seed node $q$ in graph \bm{$B$} using random walk with restart (line 2). We omit the Random Walk algorithm due to space constraints, however, the technique is described above.
    \item Determine the number of vertices to include in the subgraph \bm{$T$} based on the relevance threshold parameter $t_s$, mean of the walk distribution list $\mu(L)$ and the standard deviation of the walk distribution list $\sigma(L)$ (lines 4-6).
    \item Create a subgraph based on the included vertices (line 8).
\end{enumerate}

\begin{algorithm2e}[h]
    \KwIn{Graph \bm{$B$}, seed node $q$, teleport value $\alpha_r$, number of walks to simulate $n_w$, relevance threshold $t_s$}
    \KwResult{Subgraph \bm{$T$}}
    subgraph\_nodes = []\;
    $D$ = RandomWalk($q$, $\alpha_r$, $n_w$, \bm{$B$})\;
    $L$ = $D$.values\;

    \For{vertex $u$ in \bm{$B$}}{
        \uIf{$D[u]$ > $\mu(L)$ + $\sigma(L)$ / $t_s$}{
            subgraph\_nodes.append(u)\;
        }
    }
    \bm{$T$} = \bm{$B$}.subgraph(subgraph\_nodes)\;
    return \bm{$T$}\;
    
    \caption{Local Proximity}
    \label{NeighborhoodApproximation}
\end{algorithm2e}

    
    

\textbf{\attripart.} Armed with the \localproximity\ algorithm, we further propose an algorithm \attripart, which takes into account the network structure and attribute information contained in graph to find denser local partitions than can be found using the network structure alone. The foundation of this algorithm is based on \cite{Nibble}\cite{PageRank-Nibble}\cite{Leonid} with subtle modifications on lines 1, 4 and 9. These modifications incorporate the addition of a combined graph model, approximate PageRank computation using the \localproximity\ algorithm, and the parallel cut and conductance metric. In addition, \attripart\ doesn't depend on reaching a target conductance in order to return a local partition---instead it returns the best local partition found within sweeping $n_s$ vertices of the sorted PageRank vector.

\textit{Algorithm Description.} Given a graph \bm{$B$}, seed node $q \in V$, target conductance $\phi_o$, rank truncation threshold $\epsilon$, the number of iterations to compute the rank vector $t_{last}$, teleport value $\alpha_n$, rank iteration threshold $\epsilon_t$ and number of nodes to sweep $n_s$---\attripart\ will find a local partition $S$ around $q$ within $n_s$ iterations of sweeping. This algorithm can be viewed in five steps:

\begin{enumerate}[nolistsep]
    \item Set values for $\epsilon$ and $t_{last}$ as seen in Eq.~\eqref{epsilon_theory} and Eq.~\eqref{t_last_theory} respectively. We experimentally set $b = \frac{1 + log(m)}{2}$ and $\epsilon_t$ to 0.01. For additional detail on parameters $\epsilon$, $t_{last}$ and $b$ see \cite{Nibble}. For all other parameter values see Section \ref{ExperimentsAndDiscussion}.
    \item Run \localproximity\ around seed node $q$ in order to reduce the run time of the PageRank computations (line 1).
    \item Compute the PageRank vector using a lazy random transition with personalized restart---with preference vector $s$ containing all the probability on seed node $q$. At each iteration truncate a vertex's rank if it's degree normalized PageRank score is less than $\epsilon$ (lines 2-7). 
    \item Divide each vertex in the PageRank vector by its corresponding degree centrality and order the rank vector in descending order (line 8).
    \item Sweep over the PageRank vector for the first $n_s$ vertices, returning the best local partition $S$ found (lines 9-10). The \textit{sweep} works by taking the re-organized rank vector and creating a set of vertices $S$ by iterating through each vertex in the rank vector one at a time, each time adding the next vertex in the rank vector to $S$ and computing $\phi(S)$. 
\end{enumerate}

\begin{equation} \label{epsilon_theory}
    \epsilon = 1/(1800(l+2)t_{last}2^b)  
\end{equation}

\begin{equation}
    l=\ceil*{log_2(2m/2)}
\end{equation}

\begin{equation} \label{t_last_theory}
    t_{last} = (l+1)\ceil{\frac{2}{\phi^2}ln(c_1(l+2)\sqrt{2m/2})}
\end{equation}

\begin{algorithm2e}[h]
    \KwIn{Graph \bm{$B$}, seed node $q$, target conductance $\phi_o$, truncation threshold $\epsilon$, iterations $t_{last}$, teleport value $\alpha_n$, iteration threshold $\epsilon_t$, vertices to sweep $n_s$}
    \KwResult{Local partition $S$}
    \bm{$T$} = Local\_Proximity(\bm{$B$}, $q$, $\alpha_r$, $n_w$, $t_s$)\; 
    $\bm{D}_{i,i} = \delta(v_i)$\;
    $\bm{W} = \frac{1}{2}(\bm{I} + \bm{D}^{-1}\bm{T})$\;
    \For{$t = 1$ to $t_{last}$ \textbf{and} sum($\bm{q}_t$) - sum($\bm{q}_{t-1}$) < $\epsilon_t$}{
        $\bm{q}_t = (1 - \alpha)\bm{q}_{t-1}\bm{W} + \alpha s$\;
        $\bm{r}_t(i) = \bm{q}_t(i$) if $\bm{q}_t(i)/d(i) > \epsilon$, else 0\;
    }
 
    Order $i$ from large to small based on $\bm{r}_{t}(i)/d(i)$\;
    Sweep Parallel\_Conductance $\phi(S\{i=1..j\})$ while $i<n_s$\;
    
    If there is j : $\phi(S_j) < \phi_o$, return $S$\;
    \caption{AttriPart}
    \label{AttriPart}
\end{algorithm2e}

\textbf{\localforecasting.} As a natural application of the \attripart\ algorithm, we introduce a method to predict how local communities will evolve over time. This method is based on the \attripart\ algorithm with two significant modifications---(i) required use of the \localproximity\ algorithm to create a subgraph around the seed node and (ii) the use of the \expandedneighborhood\ algorithm to predict links between nodes in the subgraph. The idea behind using the \expandedneighborhood\ algorithm is that nodes are often missing many connections they will make in the future, which in turn affects the grouping of nodes into communities. To aid in predicting future edge connections we use Jaccard Similarity \cite{LinkPrediction} to predict the likelihood of each vertex connecting to the others---with edges added if the similarity between two nodes is greater than threshold $t_e$. 

\textit{Algorithm Description.} Given a graph \bm{$B$}, a seed node $q \in V$, a target conductance $\phi_o$, a rank truncation threshold $\epsilon$, the number of iterations to compute the rank vector $t_{last}$, a teleport value $\alpha_n$, rank iteration threshold $\epsilon_t$, similarity threshold $t_e$ and number of nodes to sweep $n_s$---this algorithm will find a {\em predicted} local partition around $q$ within $n_s$ iterations of sweeping. As the \localforecasting\ algorithm is similar to \attripart, we highlight the three primary steps:

\begin{enumerate}[nolistsep]
    \item Determine the subgraph around a given seed node using the \localproximity\ algorithm (line 1).
    \item Determine the pairwise similarity between all nodes in the subgraph using Jaccard Similarity, adding edges that are above a given similarity threshold (line 2).
    \item Run the \attripart\ algorithm to find the predicted local partition around the seed node (line 3).
\end{enumerate}

    
    
    
 
    

\begin{algorithm2e}[h]
    \KwIn{Graph \bm{$B$}, seed node $q$, target conductance $\phi_o$, truncation threshold $\epsilon$, iterations $t_{last}$, teleport value $\alpha_n$, iteration threshold $\epsilon_t$, similarity threshold $t_e$, vertices to sweep $n_s$}
    \KwResult{Predicted local partition $S$}
    
    \bm{$T$} = Local\_Proximity(\bm{$B$}, $q$)\; 
    \bm{$T$} = Expanded\_Neighborhood(\bm{$T$}, $t_e$) \;
    $S$ = AttriPart(\bm{$T$}, $q$, $\phi_o$, $\epsilon$, $t_{last}$, $\alpha_n$, $\epsilon_t$, $n_s$) \;
    return $S$ \;
    \caption{Local Forecasting}
    \label{LocalCommunityPrediction}
\end{algorithm2e}

\begin{algorithm2e}[h]
    \KwIn{Subgraph \bm{$T$}, edge addition threshold $t_e$}
    \KwResult{Subgraph \bm{$T$} with predicted edges}
    
    \For{$u$ in \bm{$T$}}{
        \For{$v$ in \bm{$T$} and $v$ not $u$}{
            u\_attr = $\bm{T}[u]$; v\_attr = $\bm{T}[v]$\;
            similarity\_score = JaccardSimilarity(u\_attr, v\_attr)\;
            \uIf{similarity\_score $>$ $t_e$ and not $\bm{T}[u][v]$}{
                $\bm{T}[u][v]$ = similarity\_score\;
            }
        }
    }
    return $T$\;
    \caption{Expanded Neighborhood}
    \label{ExpandedNeighborhood}
\end{algorithm2e}

\subsection{Analysis} \label{Algorithm_Analysis}

\paragraph{\textbf{Effectiveness}}

\localproximity\ (Algorithm~\ref{NeighborhoodApproximation}). The objective is to ensure that all relevant nodes in proximity to seed node $q$ are included. We use the fact that many real-world graphs follow a scale-free distribution \cite{ScaleFree1} \cite{ScaleFree2}, with many nodes containing only a few links while a handful encompasses the majority. In Figure \ref{RWR_Distribution}, we found that after running $n_w$ trials of random walk with restart, a scale-free like distribution formed---with a large majority of the nodes containing a small number of `hits', while a few nodes constituted the bulk.

\begin{wrapfigure}{RH}{0.25\textwidth}
\includegraphics[width=0.25\textwidth]{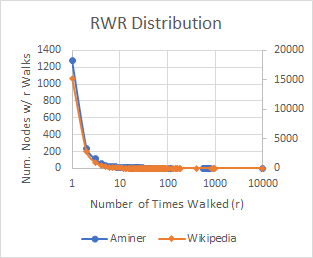}
\caption{Random walk w/ restart---distribution of node walk counts. $n_w$ = 10,000, $\alpha_r$ = 0.15; dataset: wikipedia, start vertex: `ewok', y-axis: right; dataset: Aminer, start vertex: 364298, y-axis: left. We omit nodes walked zero times in the graph, however, they're used in calculating $\mu(L)$, $\sigma(L)$.}
\label{RWR_Distribution}
\end{wrapfigure}

As the number of random walks $n_w$ is increased, the scale-free like distribution is maintained since each node is proportionally walked with the same distribution. We therefore need only some minimum value for $n_w$, which we set to 10,000. We use this skewed scale-free like distribution in combination with Eq.~\eqref{NodeRelevance} below to ensure the extraction of relevant nodes in relation to a query vertex.

Mathematically we define node relevance based on Eq.~\eqref{NodeRelevance}, where $D$ is a dictionary containing the walk count of each vertex and $D(v)$ represents the number of times vertex $v$ is walked in $n_w$ trials of the random walk with restart. $L$ is a list of each node's walk count in the graph, $\mu(L)$ is the average number of times all of the nodes in the graph are walked and $\sigma(L)$ is the standard deviation of the number of times all of the nodes in the graph are walked.  In section \ref{ExperimentsAndDiscussion} we discuss values of $t_s$ that have been shown to be empirically effective.

\begin{equation} \label{NodeRelevance}
    D(v) > \mu(L) + \sigma(L) / t_s
\end{equation}

After determining the relevant nodes we create a subgraph \bm{$T$} from a portion of the long-tail curve as defined by threshold parameter $t_s$ in conjunction with $\mu(L)$ and $\sigma(L)$. We say that subgraph \bm{$T$} contains $p \ll n$ nodes---with $p$ increasing nearly independently of the graph size (depending on threshold $t_s$). As seen in Figure \ref{RWR_Distribution} the number of nodes with $r$ walks converges independent of graph size.

\paragraph{\textbf{Efficiency}} All algorithms use the same data structure for storing the graph information. If a compressed sparse row (CSR) format is used, the space complexity is $O(2m+n+1)$. Alternatively, we note that with minor modification to the algorithms above we can use an adjacency list format with $O(n + m)$ space.


\begin{lemma}[Time Complexity]
\localproximity\ has a time complexity of $O(n+m_p+n_w)$ while \attripart\ has a time complexity of $O(p^2+pm_p+n+n_w)$ and \localforecasting\ a time complexity of $O(p^2+pm_e+n+n_w)$.
\end{lemma}

\begin{proof}
\localproximity: There are three major components to this algorithm: (1) $n_w$ random walks with walk length $l$ for a time complexity of $O(n_w)$ (line 2). (2) Linear iteration through the number of nodes taking $O(n)$ (lines 4-7). (3) Subgraph \bm{$T$} creation based on the number of included vertices $p$ with node set $V_t$---requiring iteration through every edge of node $v \in V_t$ for $m_p$ total edges. Iterating through every edge is linear in the number of edges for a time complexity of $O(m_p)$ (line 8). This leads to a total time complexity of $O(n+m_p+n_w)$

\attripart: There are six major steps to this algorithm: (1) calling \localproximity\ which returns a subgraph \bm{$T$} containing $p$ nodes and $m_p$ edges for a time complexity of $O(n+m_p+n_w)$ (line 1). (2) Creating a diagonal degree matrix by iterating through each node in \bm{$T$} with time complexity $O(p)$ (line 2). (3) Creating the lazy random walk transition matrix \bm{$W$}, which requires $O(m_p)$ from multiplying the corresponding matrix entries (line 3). (4) In lines 4-7 we iterate for $t_{last}$ iterations, with each iteration (i) updating the rank vector by multiplying the corresponding edges in the transition matrix \bm{$W$}, with the rank vector \bm{$q$} for a time complexity of $O(m_p)$ and (ii) truncating every vertex with rank $\bm{q}_t(i)/d(i) \leq \epsilon$ for a time complexity linear in the number of nodes in the rank vector $O(p)$. (5) Sort the rank vector which will be upper bounded by $O(plogp)$ (line 8). (6) Compute the parallel conductance, which takes $O(p^2+pm_p)$ time (lines 9-10). Combining each step leads to a total time complexity of $O(p^2+pm_p+n+n_w)$.

\localforecasting: This algorithm has three major steps: (1) run the \localproximity\ algorithm, which has a time complexity of $O(n+m_p+n_w)$. (2) Perform the \expandedneighborhood\ algorithm, which densifies \bm{$T$} by adding predicted edges for a total of $m_e$ edges in \bm{$T$}. This algorithm has a time complexity of $O(p^2)$ due to the nested for loops. (3) Run the \attripart\ algorithm, which has a time complexity of $O(p^2+pm_e+n+n_w)$ with the modification of $m_p$ to $m_e$ for the additional edges. This leads to an overall time complexity of $O(p^2+pm_e+n+n_w)$.
\end{proof}

While \attripart\ and \localforecasting\ both scale quadratically with respect to $p$, we note that in practice these algorithms are very fast since $p \ll n$ and $p$ scales nearly independent of graph size as shown in section \ref{Algorithm_Analysis}.

\section{Experiments} \label{ExperimentsAndDiscussion}
In this section, we demonstrate the effectiveness and efficiency of the proposed algorithms on three real-world network datasets of varying scale.

\subsection{Experiment setup}
\textbf{Datasets.} We evaluate the performance of the proposed algorithms on three datasets---(1) the Aminer co-authorship network \cite{Aminer}, (2) a Musician network mined from DBpedia and (3) a subset of Wikipedia entries in DBpedia containing both abstracts and links. All three networks are undirected with detailed information on each below:

\begin{itemize}
    \item \textbf{Aminer.} Nodes represents an author, with each author containing a set of topic keywords, and an edge representing a co-authorship. To form the attribute network, we compute attribute edges based on the similarity between two authors for every network edge, using Jaccard Similarity on the corresponding authors's topic set.
    
    \item \textbf{Musician.} Nodes represent a Musician, with each Musician containing a set of music genres, and an edge representing two Musicians who have played in the same band. To form the attribute network, we compute attribute edges based on the similarity between two Musicians for every network edge, using Jaccard Similarity on the corresponding artist's music genre set.
    
    \item \textbf{Wikipedia.} Nodes represent an entity, place or concept from Wikipedia which we will jointly refer to as an item. Each item contains a set of defining key words; with edges representing a link between the two items. The dataset originates from DBpedia as a directed graph with links between Wikipedia entries. We modify the graph to be undirected for use with our algorithms---which we believe to be a reasonable as each edge denotes a relationship between two items. In addition, this dataset uses only a portion of the Wikipedia entries containing both abstracts and links to other Wikipedia pages found in DBpedia. To form the attribute network, we compute attribute edges based on the similarity between two items for every network edge using Jaccard Similarity on the corresponding item's key word set.
\end{itemize}

\begin{table}[h!]
\begin{center}
 \begin{tabular}{|c | c | c | c |} 
 \hline
 Category & Network & Nodes & Edges \\ [0.5ex] 
 \hline\hline
 Aminer & Co-Author & 1,560,640 & 4,258,946 \\ 
 \hline
 Musician & Co-Musician & 6,006 & 8,690 \\
 \hline
 Wikipedia & Link & 237,588 & 1,130,846 \\
 \hline
\end{tabular}
\end{center}
\caption{\large{Network Statistics}}
\label{table:2}
\end{table}

\textbf{Metrics.} 
(1) To benchmark  the \localproximity\ algorithm's effectiveness and efficiency,  we compare (i)  the difference between local partition created with and without the \localproximity\ algorithm on \attripart and (ii) the run time and difference between the top 20 PageRank vector entries with and without the \localproximity\ algorithm. (2) To benchmark the \attripart\ algorithm's effectiveness and efficiency we compare the triangle count, node count, local partition density and run time to PageRank-Nibble. Normally, PageRank-Nibble does not return a local partition if the target conductance is not met, however, we modify it to return the best local partition found---even if the target conductance is not met. This modification allows for more comparable results to \attripart. (3) To provide a baseline for the \localforecasting\ algorithm's effectiveness, we compare the local partition results to \attripart\ on two graph missing 15\% of their edges.

\textbf{Repeatability.} All data and source code used in this research will be made publicly available. The Aminer co-authorship network can be found on the Aminer website \footnote{https://Aminer.org/data}; the Musician and Wikipedia datasets used in the experiments will be released on the author's website. All algorithms and experiments were conducted in a Windows environment using Python. 

\subsection{Effectiveness}
\textbf{\localproximity.} In Figure \ref{NeighborhoodResults} parts (a)-(c), we can see that the proposed \localproximity\ algorithm significantly reduces the computational run time, while maintaining high levels of accuracy across both metrics. Parts (a)-(b) demonstrate to what extent the accuracy of the results are dependent upon the parameter values. In particular, a low value of $\alpha_r$ (random walk alpha) and a high value of $t_s$ (relevance threshold) are critical to providing high accuracy results. 

In Figure \ref{NeighborhoodResults} part (a), we measure accuracy as the number of vertices that differ between the local partitions w/ and w/o the \localproximity\ algorithm on \attripart. A small partition difference indicates that the \localproximity\ algorithm finds a relevant subgraph around the given seed node and that the full graph is unnecessary for accurate results. In part (b), we define the accuracy of the results to be the difference between the set of top {\em 20} entries in the PageRank vectors for the full graph and subgraph using the \localproximity\ algorithm. Overall, the results from part (b) correlate well to (a)---showing that for low values of $\alpha_r$ (random walk alpha) and high values of $t_s$ (relevance threshold), their is negligible difference between the results computed on the full graph and the subgraph found using the \localproximity\ algorithm.

\begin{figure}
  \begin{subfigure}{\linewidth}
  \includegraphics[width=\linewidth]{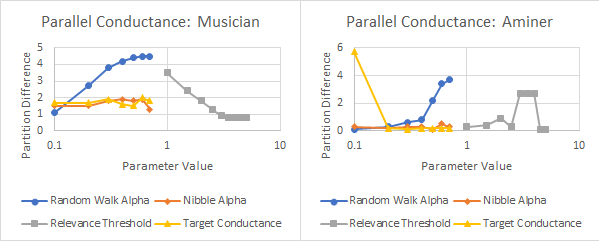}\hfill
  \caption{Y-axis represents the difference in vertices between the local partition calculated w/ and w/o the \localproximity\ algorithm.}
  \end{subfigure}\par\medskip
  \begin{subfigure}{\linewidth}
  \includegraphics[width=\linewidth]{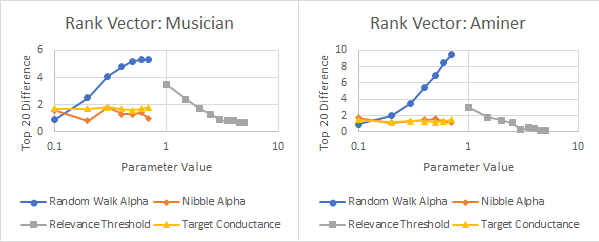}\hfill
  \caption{Y-axis represents the \# of vertices differing between the top 20 rank vector entries w/ and w/o the \localproximity\ algorithm.}
  \end{subfigure}
  \begin{subfigure}{\linewidth}
  \includegraphics[width=\linewidth]{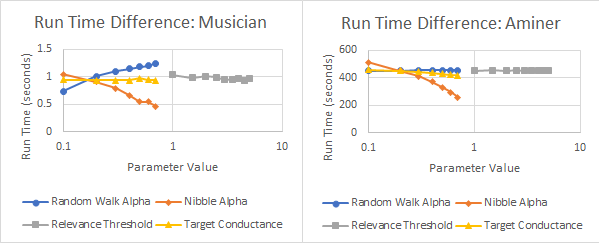}\hfill
  \caption{Y-axis represents the difference in run time between the PageRank calculation w/ and w/o the \localproximity\ algorithm.}
  \end{subfigure}
  \caption{Each data point averages 10 randomly sampled vertices in both the Aminer and Musician datasets. Default parameters (unless sweeped across): $\alpha_n$ =  0.2, $\alpha_r$ = 0.15, $\phi_o$ = 0.2, $t_s$ = 2, $n_w$ = 10,000, $n_s$ = 200. Parameter ranges: $\alpha_r$, $\alpha_n$ and $\phi_o$ [0.1-0.7] in 0.1 intervals; $t_s$ [1-5] in 0.5 intervals.}
  \label{NeighborhoodResults}
\end{figure}


\textbf{\attripart.} In Figure \ref{LocalPartitionResults}, we see that \attripart\ finds significantly denser local partitions than PageRank-Nibble---with local partition densities approximately \textbf{1.6$\times$}, \textbf{1.3$\times$} and \textbf{1.1$\times$} higher in \attripart\ than PageRank-Nibble in the Aminer, Wikipedia and Musician datasets respectively. Density is measured as $\frac{2m}{n(n-1)}$ where $m$ is the number of edges and $n$ is the number of nodes. 

In Figure \ref{LocalPartitionResults}, we observe that the triangle count of the \attripart\ algorithm is lower than PageRank-Nibble in the Musician and Aminer datasets. We attribute this to the fact that \attripart\ is finding smaller partitions (as measured by node count) and, therefore, there are less possible triangles. We also note that each triangle is counted three times, once for each node in the triangle. While no sweeps across algorithm parameters were performed, we believe that the gathered results provide an effective baseline for parameter selection. 

\begin{figure} 
  \begin{subfigure}{\linewidth}
  \includegraphics[width=\linewidth]{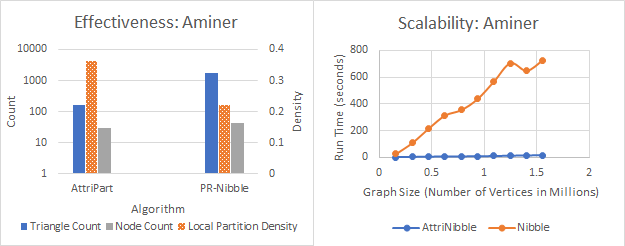}
  \caption{Scalability: Each data point represents the Aminer dataset in 1/10th intervals, with each point averaged over 3 randomly sampled vertices. Parameters: $\alpha_n$ =  0.2, $\alpha_r$ = 0.15, $\phi_o$ = 0.2, $t_s$ = 2, $n_w$ = 10,000, $n_s$ = 200.}\hfill
  \end{subfigure}\par\medskip
  
  \begin{subfigure}{\linewidth}
  \includegraphics[width=\linewidth]{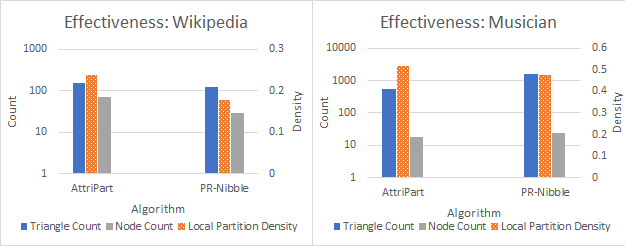}\hfill
  \caption{}
  \end{subfigure}\par\medskip
  \caption{Effectiveness: results are averaged over 20 and 100 randomly sampled vertices in the Aminer/Wikipedia and Musician datasets, respectively. Parameters: $\alpha_n$ =  0.2, $\alpha_r$ = 0.15, $\phi_o$ = 0.05, $t_s$ = 2, $n_w$ = 10,000, $n_s$ = 200.}
  \label{LocalPartitionResults}
\end{figure}

\textbf{\localforecasting.} In order to measure the effectiveness of the \localforecasting\ algorithm we setup the following experiment with three local partition calculations: (1) calculate the local partition using \attripart, (2) calculate the local partition using \attripart\ with 15\% of the edges randomly removed from the graph and (3) calculate the local partition using the \localforecasting\ algorithm with 15\% of the edges randomly removed from the graph. We treat (1) as the baseline local community and want to test if (3) finds better local partitions than (2). The idea behind randomly removing 15\% of the edges in the graph is to simulate the evolution of the graph over time and test if the \localforecasting\ algorithm can predict better local communities in the future. Ideally, we would have ground-truth local community data for a rich graph with time series snapshots, however, in its absence we use the above method.  

In Figure \ref{PredictionEffectivness}, each data point is generated in three steps---(i) taking the difference between the set of vertices and edges in local partitions (1) and (3), (ii) taking the difference between the set of vertices and edges in local partitions (1) and (2) and (iii) by taking the difference between (ii) and (i). Step (i) tells us how far off the \localforecasting\ algorithm is from the baseline, step (ii) tells us how far off the local partition would be from the baseline if no prediction techniques were used and step (iii) tells us the difference between the local partitions with and without the \localforecasting\ algorithm (which is what we see graphed in Figure \ref{PredictionEffectivness}).

In Figure \ref{PredictionEffectivness}, we see that the local partition prediction accuracy, for both the edges and vertices, is above the baseline calculations in the Aminer dataset for a majority of edge similarity threshold values ($t_e$). The best results were obtained when $t_e$ is 0.6, with an average of 1.4 vertices and 2.75 edges predicted over the baseline using the \localforecasting\ algorithm. This number, while relatively small, is an average of 20 randomly sampled vertices---with one result reaching up to 14 vertices and 26 edges over baseline. In addition, we can see that the Musician dataset does not perform as well as the Aminer dataset, with most of the prediction results performing worse than the baseline (as indicated by the negative difference). We believe that this result on the Musician dataset is due to the different nature of each dataset's network structure---with the Musician dataset being significantly more sparse (no giant connected component) than the Aminer dataset. 

\begin{figure}[h!]
\includegraphics[height=3.34cm, width=8.5cm]{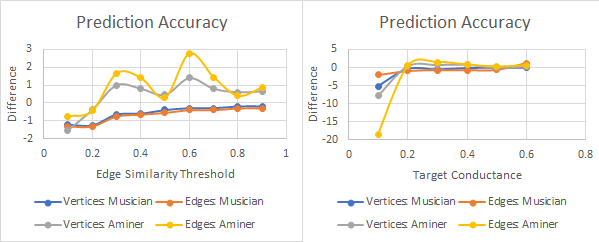}
\caption{Each data point averages 20 randomly sampled vertices in the Aminer and Musician datasets. Default parameters (unless sweeped across): $\alpha_n$ =  0.2, $\alpha_r$ = 0.15, $\phi_o$ = 0.2, $t_s$ = 5, $t_e$ = 0.7, $n_w$ = 10,000, $n_s$ = 200. Parameter ranges: $t_e$ [0.1-0.9] in 0.1 intervals, $\phi_o$ [0.1-0.6] in 0.1 intervals.}
\label{PredictionEffectivness}
\end{figure}


\subsection{Efficiency} For both the proposed and baseline algorithms, the efficiency results represent only the time taken to run the algorithm (e.g. not including loading data into memory). \textbf{\localproximity.} Across a majority of the parameters the run time for the full graph PageRank computation is approximately 450 seconds longer compared to computing the PageRank vector based on the \localproximity\ sugraph. \textbf{\attripart.} In Figure \ref{LocalPartitionResults}, we see that the \attripart\ algorithm finds local partitions \textbf{43$\times$} faster than PageRank-Nibble. \textbf{\localforecasting.} This algorithm has an expected run time nearly identical to \attripart, we therefore refer the reader to Figure \ref{LocalPartitionResults} for run time results.  

\section{Related Work}

We provide a high level review of both local and global community detection methods, with a focus on the research that pertains to the algorithms we propose in this paper. 

\noindent {\bf A - Local Community Detection.} \label{NibbleExplained}
Given an undirected graph, start vertex and a target conductance---the goal of Nibble is to find a subset of vertices that has conductance less than the target conductance \cite{Nibble}. This algorithm has strong theoretical properties with a run time of $O(2^b(log^6m)/\phi^4)$, where $b$ is a user defined constant, $\phi$ is the target conductance and $m$ is the number of edges. {PageRank-Nibble} builds on the work of Nibble by introducing the use of personalized PageRank \cite{PersonalizedPageRank,Tong06@EUROSIP}, in addition to an algorithm for the computation of approximate PageRank vectors \cite{PageRank-Nibble}. Since PageRank-Nibble and Nibble run on undirected graphs, they use truncated random walks in order to prevent the stationary distribution from becoming proportional to the degree centrality of each node \cite{UndirectedPageRank}. There are also many {alternative techniques} for local community detection. To name a few, the paper by Bagrow and Bollt \cite{LocalCommunity} introduces a method of local community identification that utilizes an $l$-shell spreading outward from a start vertex. However, their algorithm requires knowledge of the entire graph and is therefore not truly local. The research by J. Chen et. al. \cite{LocalSocial} proposes a method for local community identification in social networks that avoids the use of hard to obtain parameters and improves the accuracy of identified communities by introducing a new metric. In addition, the work by \cite{LocalClusterHighOrderNetwork} and \cite{HighOrderClustering} introduces two methods of local community identification that take into account high-order network structure information. In \cite{LocalClusterHighOrderNetwork}, the authors provide mathematical guarantees of the optimality and scalability of their algorithms, in addition to the generalization of it to various network types (e.g. signed and multi-partite networks).

\noindent {\bf B - Global Community Detection.} The basic idea behind the Walktrap algorithm is that random walks on a graph tend to get "trapped" in densely connected parts that correspond to communities \cite{Walktrap}. Utilizing the properties of random walks on graphs, they define a measurement of structural similarity between vertices and between communities, creating a distance metric. The algorithm itself has an upper bound of $O(mn^2)$. Another popular choice for global community detection is spectral analysis. In the paper by M. Newman \cite{NewmanSpectral} it is shown that the problems of community detection by modularity maximization, community detection by statistical inference and normalized-cut graph partitioning when tackled using spectral methods, are in fact, the same problem. The work by S. White et. al. in \cite{Spectral1} attempts to find communities in graphs using spectral clustering. They achieve this by using an objective function for graph clustering \cite{NewGir04} and reformulating it as a spectral relaxation problem, for which they propose two algorithms to solve it. A systematic introduction to spectral clustering techniques can be found in \cite{SpectralTutorial}. There also exists many alternative techniques for global community detection. Among others, two interesting techniques relevant to this work are \cite{CommunityDetectionAttributes} \cite{CommunityPrediction}. In \cite{CommunityDetectionAttributes}, the authors propose a community detection algorithm that uses the information in both the network structure and the node attributes, while in \cite{CommunityPrediction} the authors use network feature extraction to predict the evolution of communities. A detailed review of various community detection algorithms can be found in \cite{ClusteringOverview}.  




\section{Conclusion}
This paper proposes new algorithms for attributed graphs, with the goal of (i) computing denser local graph partitions and (ii) predicting the evolution of local communities. We believe that the proposed algorithms will be of particular interest to data mining researchers given the computational speed-up and enhanced dense local partition identification. The proposed local partitioning algorithm \attripart\ has already \textbf{deployed} to the web platform \href{www.path-finder.io}{PathFinder} (www.path-finder.io) \cite{RapidAnalysisNetworkConnectivity} and allow users to interactively explore all three datasets presented in the paper. In addition, the source code and datasets will be made publicly available by the conference date.



\bibliographystyle{ACM-Reference-Format}
\bibliography{sample-bibliography}

\end{document}